\documentclass[12pt]{amsart}

\usepackage{amssymb,amsmath,amsfonts,amssymb,bm,stmaryrd}
\usepackage{graphicx,color,colortbl,texdraw,enumerate}
\usepackage{calrsfs}
\usepackage[round]{natbib}
\usepackage[parfill]{parskip}    
\allowdisplaybreaks

\definecolor{labelkey}{rgb}{0.6,0,1}

\numberwithin{equation}{section}

\renewcommand{\d}{{\rm d}}
\newcommand{\e}{{\rm e}}
\newcommand{\E}{{\mathbb E}}

\renewcommand{\P}{{\mathbb P}}
\newcommand{\Q}{{\mathbb Q}}

\newcommand{\R}{{\mathbb R}}

\newcommand{\Acal}{{\mathcal A}}

\newcommand{\Ecal}{{\mathcal E}}
\newcommand{\Fcal}{{\mathcal F}}
\newcommand{\Gcal}{{\mathcal G}}

\DeclareMathOperator{\tr}{trace}

\newcommand{\1}[1]{{\boldsymbol 1_{\{#1\}}}}
\newcommand{\oo}{{\bf 1}}

\newtheorem{theorem}{Theorem}

\newtheorem{assumption}{Assumption}

\newtheorem{corollary}{Corollary}

\newtheorem{definition}{Definition}
\newtheorem{example}{Example}

\newtheorem{lemma}{Lemma}

\newtheorem{proposition}{Proposition}

\begin{document}

\title[]{Non-Equivalent Beliefs and\\Subjective Equilibrium Bubbles}
\author[M.~Larsson]{Martin Larsson}
\address{Martin Larsson\\Swiss Finance Institute\\ 
Ecole Polytechnique F\'ed\'erale de Lausanne\\
Lausanne, Switzerland}
\email{martin.larsson@epfl.ch}
\date{\today}

   

\thanks{Swiss Finance Institute, Ecole Polytechnique F\'ed\'erale de Lausanne. E-mail: martin.larsson@epfl.ch}
\thanks{I would like to thank Michalis Anthropelos, Julien Hugonnier, Robert Jarrow, Semyon Malamud, Loriano Mancini, and Paul Schneider for very stimulating discussions and valuable comments.}

\begin{abstract}
This paper develops a dynamic equilibrium model where agents exhibit a strong form of belief heterogeneity: they disagree about zero probability events. It is shown that, somewhat surprisingly, equilibrium exists in this setting, and that the disagreement about nullsets naturally leads to equilibrium asset pricing bubbles. The bubbles are subjective in the sense that they are perceived by some but not necessarily all agents. In contrast to existing models, bubbles arise with no restrictions on trade beyond a standard solvency constraint.
\end{abstract}

\maketitle

\section{Introduction}

Since the work of \citet{Santos/Woodford:1997} and \citet{Loewenstein/Willard:2000,Loewenstein/Willard:2006} it is well understood that, under fairly broad economic assumptions, the equilibrium price of an asset in positive net supply must equal its fundamental value, defined as the smallest replication cost of the associated cash flows. Such models can therefore not be used to describe rational asset pricing bubbles, defined as a wedge between the market price and the fundamental value, as the result of an equilibrium mechanism. On the other hand, the occurrence of bubbles in real markets makes it desirable to develop equilibrium models that can accommodate them, and within which their properties can be analyzed. To date this has been achieved by introducing market frictions, such as trading constraints, into the standard setting. This prevents investors from exploiting or scaling up arbitrage opportunities, and allows bubbles to persist (see below for a review of this literature.)

The present article deviates from the classical framework in a different manner. Instead of being subjected to trading constraints, the agents are endowed with heterogeneous beliefs. This heterogeneity is taken to a rather extreme degree: agents do not even agree about zero probability events. It is far from clear from the outset that an equilibrium can exist in such a setting; one objection will be discussed momentarily. Our first contribution is to show that, somewhat surprisingly, this is indeed possible. Secondly, once the equilibrium has been found, it is shown that asset pricing bubbles arise naturally.

Regarding existence of equilibrium, a simple (but incorrect) argument can be given to suggest that matters are hopeless. Suppose an equilibrium with two agents is given, where some time~$T$ event $D$ (``downturn'') is assigned positive probability by agent~1, but zero probability by agent~2. Agent~1 is then willing to purchase insurance against $D$, say by buying an indicator asset $\oo_D$ at a positive price. Agent~2 is happy to supply this asset, since his valuation of it is zero. But as the price is positive, agent~2 would like to scale up his position. This violates individual optimality and contradicts the existence of equilibrium.

This argument has a key flaw. Because, in order to implement the suggested transaction, the agents must trade dynamically over time to ensure that the correct amount is delivered at time~$T$. Agent~1 believes this amount will be either zero or one, whereas agent~2 is convinced the amount will be zero. Nonetheless, the position of agent~2 may run intermediate deficits, which prevents him from scaling up the position indefinitely without risking insolvency. This is analogous to the well-known mechanism that allows bubbles to exist without violating arbitrage restrictions. Of course, agent~2 will scale his position as much as possible, which means that under certain scenarios he will come very close to insolvency---and if $D$ does, in fact, occur (even though agent~2 thought it impossible), he will go bankrupt. It turns out that this last statement holds in full generality, and is what allows us to obtain a consistent equilibrium; it enables us to deal with the otherwise uncomfortable question of what the agent does when an event he thought was impossible nonetheless occurs. A detailed discussion of this point is given in Section~\ref{S:nores}.

Once in equilibrium, the mechanism leading to bubbles can be understood as follows. An agent computes the fundamental value of an asset as the smallest amount needed to replicate its cash flows \emph{almost surely}. However, when there is disagreement about nullsets, the notion of ``almost sure replication'' varies between different agents. As a result it can happen that an agent assigns zero probability to some scenarios that other agents believe are possible. In this case, he does not have to replicate the cash flows in those scenarios, and is therefore able to carry out the replication at a lower cost. This leads to disagreement about fundamental values, and thus automatically to bubbles, since all agents must agree on the (endogenous and observed) market price. Both bubbles and fundamental values are thus \emph{subjective} in that they depend on the beliefs of the agent computing them. It turns out that this dependence is only through the nullsets. In particular this implies that no bubbles will be present if all agents agree about nullsets, even if beliefs are otherwise heterogeneous.

As soon as equivalence of beliefs is dropped, a range of new phenomena appears. We analyze some of these by means of various explicit examples. The ``downturn'' situation discussed above is formalized in a setting where one agent is optimistic in the sense that he believes the stock dividend cannot fall below some given level (see Section~\ref{S:ex1}). This agent will then try to exploit a perceived bubble by shorting the riskless asset and holding the stock as collateral to avoid insolvency. If the dividend process nonetheless falls below the given level, the agent goes bankrupt. The presence of bubbles implies that no equivalent martingale measure exists relative to this agent's beliefs. However, the other agent does not see any bubbles, and for him an equivalent (true) martingale measure can be found. Section~\ref{S:ex2} contains a variation on the same theme. We also study an economy with two agents, where \emph{both agents simultaneously perceive bubbles} (see Section~\ref{S:ex3}). In this setting the bubble can either persist until one of the agents goes bankrupt, or burst at an earlier point in time. A fourth example (Section~\ref{S:ex4}) is concerned with a two-agent, two-stock economy, where both agents perceive bubbles. Here an additional interesting phenomenon occurs: not only do the two agents agree about the existence of a bubble, they also agree about the statistical properties of the bubble on the market portfolio, in the sense that its (unconditional) law is the same under the two agents' beliefs.

Let us give a brief overview of some related literature. The paper most closely related to ours is~\citep{Hugonnier:2012}, where it is shown that rational asset pricing bubbles can arise in equilibria where some, but not all, agents face restrictions on trade. In a similar vein, there are several papers where portfolio constraints lead to mispricings, for instance \citep{Cuoco:1997,Basak/Cuoco:1998, Basak/Croitoru:2000}. Without such constraints, and with homogenous beliefs, bubbles are not possible as long as agents must maintain nonnegative wealth, see~\citep{Loewenstein/Willard:2000}. In the present paper, bubbles arise without any restrictions on trade, other than the nonnegative wealth restriction.

Another line of research considers exogenously specified arbitrage-free models for asset prices. Here bubbles (or relative arbitrage) occur by construction, without any equilibrium mechanism justifying their presence. Representative papers include \citep{Sin:1998, Heston/Loewenstein/Willard:2007, Jarrow/Protter/Shimbo:2006, Jarrow/Protter/Shimbo:2010, Fernholz/Karatzas/Kardaras:2005, Fernholz:2010fk, Ruf:2011}. Our results have consequences for the interpretation of this type of models, since bubbles are subjective in our setting: some agents perceive them, while others do not. The appearance of bubbles in a given asset pricing model could therefore in principle be due to a misspecification of the historical probability measure, assigning zero probability to scenarios (such as defaults or market crashes) that may in fact not be impossible events. The present paper should be viewed as a contribution toward understanding to what extent such exogenously specified pricing models are supported by equilibrium price formation mechanisms.

A third line of research is a large body of work dealing with general equilibrium models with heterogeneous beliefs, a small sample of which is \citep{Basak:2000,Basak:2005,Berrada/Hugonnier/Rindisbacher:2007,Jouini/Napp:2007,Cvitanic/etal:2012}. However, all these papers assume that agents agree about zero probability events, with the consequence that no bubbles occur. Non-equivalent beliefs occur in \citep{Epstein/Ji:2013} in the context of ambiguous volatility, although the setting of that paper is very different from ours.


The rest of this paper is structured as follows. Section~\ref{S:NE} describes the model setup and the resulting equilibrium. An important issue in the context of non-equivalent beliefs is that some agents may become insolvent in finite time. This leads us to introduce the notion of \emph{No~Resurrection}. Section~\ref{S:B} discusses the occurrence of bubbles, in particular making precise the concept of subjective bubbles. Section~\ref{S:ex} is devoted to examples, and Section~\ref{S:concl} concludes. All proofs are deferred to the~Appendix.

Throughout this paper we work on a probability space $(\Omega,\mathcal F, \mathbb F, \P)$ equipped with a filtration $(\Fcal_t)_{0\le t\le T}$ for a bounded time set $[0,T]$, generated by some $n$-dimensional Brownian motion $X=(X_1,\ldots,X_n)$. Expectation under $\P$ is denoted $\E[\,\cdot\,]$, and we use the shorthand notation $\E_t[\,\cdot\,]=\E[\,\cdot \mid\Fcal_t]$. For integrable functions $f$ we write $\int_a^b f(s)\d s=\int_{[a,b)}f(s)\d s$. In particular, this quantity is equal to zero whenever~$b\le a$.

\section{Equilibrium with non-equivalent beliefs} \label{S:NE}

We consider a multi-agent economy with heterogeneous beliefs, where the agents may disagree about zero-probability events. To a large extent this material parallels well-known developments. However, there are certain delicate points arising from the non-equivalence of beliefs that will be commented on in some detail. It should be emphasized that the belief heterogeneity is the only point where our model differs from standard dynamic equilibrium models with intermediate consumption as described, for instance, in \citep{Duffie:2001}.


\subsection{Agents and preferences}

Finitely many agents, indexed by $k=1,2,\ldots,K$, maximize expected utility from consumption over the time interval $[0,T]$. The preferences of agent~$k$ are given by his beliefs, determined by a probability measure $\P_k\ll \P$, a time preference rate $\rho>0$ (which for simplicity we take to be the same for all agents), and a utility function $u_k:(0,\infty)\to\R$. It is assumed that $u_k$ is continuously differentiable, strictly increasing, strictly concave, and satisfies Inada conditions at zero and infinity. These assumptions are standard. The utility agent $k$ achieves from a consumption plan $c=(c_t)_{0\le t\le T}$ is given by
\[
U_k(c) = \E^k\left[ \int_0^T e^{-\rho t} u_k(c_t) \d t \right],
\]
whenever the right side is well-defined, and $-\infty$ otherwise. Here $\E^k[\,\cdot\,]$ denotes expectation under $\P_k$. This preference structure is standard, with the exception that \emph{$\P \ll \P_k$ is not assumed}. Similarly, we do not impose any absolute continuity relationship between $\P_k$ and $\P_\ell$ for $k\ne\ell$. This is what drives all the subsequent results on equilibrium bubbles.

Before describing the consumption plans and investment strategies available to the agents, we establish some notation regarding the agents' beliefs. Let
\[
Z_{kt} = \E_t\left[ \frac{\d \P_k}{\d \P} \right], \qquad 0\le t\le T,
\]
be the density process associated with $\P_k$, and define
\[
\tau_k = \inf\left\{t\in [0, T]: Z_{kt} = 0\right\}.
\]
(As usual, we set $\inf \emptyset = \infty$.) While it is by no means clear at this stage, it will be shown later on that $\tau_k$ corresponds to a bankruptcy time of agent~$k$. It is easy to verify that $\P(\tau_k\le T)>0$ holds if and only if $\P_k$ and $\P$ are not equivalent, and that
\begin{equation} \label{eq:locequiv}
\P_k\mid_{\Fcal_t\cap\{t<\tau_k\}}\ \sim\ \P\mid_{\Fcal_t\cap\{t<\tau_k\}} \text{ holds for all } t\in[0,T].
\end{equation}
Heuristically this means that $\P_k$ and $\P$ are ``equivalent strictly prior to $\tau_k$''. This will often let us deduce that various properties of interest hold for all $0\le t\le T$, $\P_k$-a.s., if and only if they hold for all $0\le t<\tau_k$, $\P$-a.s. Next, notice that we have
\begin{equation} \label{eq:Pk0}
\P_k(\tau_k\le T) = \E\left[ Z_{kT}\1{\tau_k\le T}\right] = 0.
\end{equation}
Due to the interpretation alluded to above of~$\tau_k$ as a bankruptcy time, this simple means that agent~$k$ believes that he will never become bankrupt. The following assumption will always be in force.

\begin{assumption} \label{A:max}
We have $\P(\max_k \tau_k = \infty)=1$.
\end{assumption}

Given the above interpretation of $\tau_k$, this assumption means that there is always some agent who remains solvent up to time~$T$. Note that this does not exclude situations where $\P(\tau_k\le T)>0$ for all~$k$, as long as in every state of world, there is some $k$ (depending on the state) for which $\tau_k=\infty$. Assumption~\ref{A:max} is not a real restriction, since it can be achieved by replacing $\P$ by $(\P_1+\cdots+\P_K)/K$.

\subsection{Assets and strategies} \label{S:NE1}

We now describe the consumption plans and investment strategies available to the agents. The framework is again standard, apart from the disagreement about nullsets. The financial market consists of $n$ stocks in unit net supply, and a money market account in zero net supply. The $i$:th stock is a claim on a dividend rate process $D_i=(D_{it})_{0\le t\le T}$, whose dynamics is given by
\[
D_{it} = D_{i0} + \int_0^tD_{is} a_{is} \d s + \int_0^t D_{is} v_{is}^\top \d X_s,
\]
for an exogenously given $\mathbb R$-valued drift $a_i=(a_{it})_{0\le t\le T}$ and $\mathbb R^n$-valued volatility $v_i=(v_{it})_{0\le t\le T}$, assumed to satisfy the integrability condition
\begin{equation} \label{eq:int_e}
\int_0^T \left( |a_{it}| + \| v_{it} \|^2 \right) \d t < \infty \quad \P\text{-a.s.}
\end{equation}
This implies that the dividend processes are strictly positive. So is the aggregate dividend rate $D=(D_t)_{0\le t\le T}$, defined by $D_t=D_{1t}+\cdots+D_{nt}$, which consequently can be written
\begin{equation} \label{eq:D}
D_t = \sum_{i=1}^n D_{it} = D_0 + \int_0^t D_s a_s \d s + \int_0^t D_s v_s^\top \d X_s
\end{equation}
for some drift $a=(a_t)_{0\le t\le T}$ and volatility vector $v=(v_t)_{0\le t\le T}$.

The market prices $S=(S_{1t},\ldots,S_{nt})_{0\le t\le T}$ of the stocks, as well as the risk-free rate $r=(r_t)_{0\le t\le T}$, are determined endogenously in equilibrium; the former within the class of It\^o processes of the form
\begin{equation} \label{eq:Si}
S_{it} + \int_0^t D_{is} \d s = S_{i0} + \int_0^t S_{is}\mu_{is} \d s + \int_0^t S_{is} \sigma_{is}^\top \d X_s, \qquad 0\le t\le T,
\end{equation}
with the integrability condition
\[
\int_0^T \left(|r_r| +  \| \mu_t \| + \tr (\sigma_t\sigma_t^\top) \right) \d t  < \infty \quad \P\text{-a.s.}
\]
Here $\mu_t=(\mu_{1t},\ldots,\mu_{nt})$, and $\sigma_t$ is the matrix with rows $\sigma_{it}^\top$. We let
\begin{equation} \label{eq:S0}
S_{0t} = \exp\left( \int_0^t r_s ds \right)
\end{equation}
denote the value of the money market account. The market portfolio is denoted
\[
\overline S_t = S_{1t} + \cdots + S_{nt}.
\]
We focus on equilibria where the market is complete. In the present setting where the number of Brownian motions coincides with the number of risky assets, this corresponds to~$\sigma_t$ being invertible $\P\otimes\d t$-a.e. This leads to a unique market price of risk given by
\[
\theta_t = \sigma_t^{-1}(\mu_t - r_t{\bf 1}), \qquad \int_0^T \|\theta_t\|^2 \d t < \infty \quad \P\text{-a.s.}
\]
The associated (possibly non-normalized) state price density is
\begin{equation}\label{eq:SPD}
\xi_t = \xi_0 \exp\left(-\int_0^t r_s \d s -\int_0^t \theta_s^\top \d X_s - \frac{1}{2}\int_0^t \|\theta_s\|^2 \d s \right),
\end{equation}
whose characteristic property is that the deflated cum-dividend stock prices are local martingales; Indeed, by It\^o's formula,
\begin{equation} \label{eq:defl}
\xi_t S_{it} + \int_0^t \xi_s D_{is} \d s \quad \text{is a local $\P$-martingale.}
\end{equation}
While it will turn out that limited arbitrage opportunities may exist in equilibrium, we can always safely assume that a market price of risk process and state price density as above exist, since this is a minimal condition under which the agents' utility maximization problems will have optimal solutions with finite optimal value; see \citep{Karatzas/Shreve:1998, Duffie:2001}, as well as \citep{Karatzas/Kardaras:2007} for a thorough discussion. We write $\xi_{kt}$ for the $k$:th agent's state price density, which is only defined prior to~$\tau_k$. It is given by
\begin{equation} \label{eq:xik}
\xi_{kt} = \frac{\xi_t}{Z_{kt}}, \qquad 0\le t<\tau_k.
\end{equation}
Due to~\eqref{eq:locequiv} this implies that it is defined for all $0\le t\le T$, $\P_k$-a.s. In view of \eqref{eq:defl}, Bayes' rule (see Lemma~\ref{L:Bayes} in the Appendix) implies that $\xi_{kt}S_{it} + \int_0^t\xi_{ks}D_{is}\d s$ is a local $\P_k$-martingale.

A \emph{trading strategy} is a pair $(\phi,\pi)=(\phi_t,\pi_t)_{0\le t\le T}$ of predictable processes, where $\phi_t$ is the amount invested in the money market account at time $t$, while $\pi_t=(\pi_t^1,\ldots,\pi_t^n)$ gives the amounts invested in the stocks. These processes are required to satisfy
\[
\int_0^T \left( |\phi_sr_s + \pi_s^\top \mu_{ks}|  + \| \sigma_s^\top \pi_s \|^2 \right) \d s < \infty \quad \P_k\text{-a.s.},
\]
where $\mu_k$ is the drift vector of the stocks under~$\P_k$, obtainable via Girsanov's theorem\footnote{By \citep[Theorem~III.41]{Protter:2005}, it is given in terms of $\mu$, $\sigma$, and $Z_k$ by $\mu_{kt} = \mu_t + Z_{kt}^{-1} \sigma_t \frac{\d \langle X,Z_k\rangle_t}{\d t}$ for $t\in[0,T]$, $\P_k$-a.s.}
(note that the volatility is unaffected by absolutely continuous changes of probability measure.) If the associated wealth process is nonnegative,
\[
W_t(\phi,\pi) = \phi_t + \oo^\top \pi_t \ge 0, \quad 0\le t\le T, \quad \P_k\text{-a.s.},
\]
the strategy is called \emph{$\P_k$-admissible}. Given a consumption plan~$c\ge0$, a time point $t_0\in[0,T]$, and an $\Fcal_{t_0}$-measurable random variable $W_{t_0}$ (the time~$t_0$ wealth), a strategy $(\phi,\pi)$ is called \emph{$\P_k$-self-financing} if its wealth process satisfies\footnote{An equivalent condition is obtained by replacing $\mu$ by $\mu_k$, and $X$ by $X^k=X-\int Z^{-1}_s \d\langle X,Z_k\rangle$, which is $n$-dimensional Brownian motion under $\P_k$.}
\begin{equation} \label{eq:sfk}
W_t(\phi,\pi) = W_{t_0} + \int_{t_0}^t (\phi_sr_s + \pi_s^\top \mu_s - c_s) \d s + \int_{t_0}^t \pi_s^\top \sigma_s \d X_s, \quad t_0\le t\le T, \quad \P_k\text{-a.s.}
\end{equation}
A consumption plan for which such a strategy $(\phi,\pi)$ exists is called \emph{$\P_k$-feasible (given time~$t_0$ wealth $W_{t_0}$)}, and is said to be \emph{financed} by $(\phi,\pi)$.\footnote{One could replace the admissibility condition by suitable integrability requirements that guarantee that $\xi_{kt}W_t(\phi,\pi)+\int_0^t\xi_{ks} W_s(\phi,\pi)\d s$ is a martingale under $\P_k$. However, as discussed in \citep[Footnote~7]{Hugonnier:2012}, the resulting equilibrium is not affected by such a modification.} Note that above notions really depend on $\P_k$ in general. However, this dependence is only through the nullsets, and is not present if all $\P_k$ are equivalent. In this case all agents have access to the same set of strategies. The following (essentially standard) result characterizes the class of $\P_k$-feasible consumption plans in terms of a static budget constraint.

\begin{proposition} \label{P:budget}
$(i)$ A  consumption plan $c=(c_t)_{0\le t\le T}$ is $\P_k$-feasible given time~$t_0$ wealth $W_{t_0}$ if and only if
\[
\frac{1}{\xi_{kt_0}}\E^k_{t_0} \left[ \int_{t_0}^T \xi_{ks}c_s\d s\right] \le W_{t_0},\quad \P_k\text{-a.s.}
\]
$(ii)$ Let $C_T$ be a nonnegative $\Fcal_T$-measurable random. There exists a $\P_k$-self-financing strategy (with zero consumption and time~$t_0$ wealth $W_{t_0}$) such that its wealth process $W$ satisfies $W_T=C_T$ if and only if
\[
\frac{1}{\xi_{kt_0}}\E^k_{t_0} \left[ \xi_{kT}C_T \right] \le W_{t_0},\quad \P_k\text{-a.s.}
\]
\end{proposition}

The $k$:th agent is endowed with some positive initial wealth $w_k>0$, which comes in the form of $\alpha_{k0}$ units of the riskless asset, and $\alpha_k=(\alpha_{k1},\ldots,\alpha_{kn})$ units (nonnegative) of the stocks. To ensure compatibility with market clearing, we require $\alpha_{10}+\cdots+\alpha_{K0}=0$, $\alpha_{1i}+\cdots+\alpha_{Ki}=1$, $i=1,\ldots,n$. The agent maximizes his expected utility $U_k(c)$ over all $\P_k$-feasible consumption plans $c$, given initial wealth~$w_k$. In view of Proposition~\ref{P:budget}, the $k$:th agent's optimization problem is:
\[
\max_{c\ge 0} \ \E^k\left[ \int_0^T e^{-\rho t} u_k(c_t) \d t \right] \quad \text{s.t.} \quad \frac{1}{\xi_{k0}}\E^k\left[\int_0^T \xi_{kt}c_t\d t\right] \le w_k.
\]

\subsection{Individual optimality and No Resurrection} \label{S:nores}

Two important issues that arise in the context of non-equivalent beliefs are addressed in this subsection.

The first concern is the following. Since $\P_k(\tau_k\le T)=0$, the $k$:th agent cannot distinguish between consumption plans that only differ for $t\ge\tau_k$ (and similarly for trading strategies). The resulting non-uniqueness of optimizers is troubling: if $\tau_k<T$, the economic activities continue after $\tau_k$, and the actions of the $k$:th agent do influence the equilibrium. This necessitates some rule for determining the actions of agent~$k$ after~$\tau_k$.

The need for such a rule is reinforced by our second concern, which arises from the following fact: The optimal consumption plan is fixed by the agent at time~0, when he does not believe $\tau_k$ can ever occur. If nonetheless $\tau_k$ does occur, it seems unsatisfactory to force the agent to insist on his belief that $\tau_k$ cannot happen. It would be more reasonable for him to revise his beliefs, and continue trading. But how should these new beliefs be determined? And even if a compelling choice is possible, why did the agent from the outset fail to anticipate the potential need for these new beliefs?

Fortunately there is a simple and natural way to resolve these issues. It relies on the following result, which characterizes the optimal solution to the utility maximization problem faced by agent~$k$.

\begin{proposition} \label{P:opt}
An optimal $\P_k$-feasible consumption plan for the $k$:th agent's optimization problem is given by
\begin{equation} \label{eq:FOCk}
c_{kt} = I_k\left(y_k\xi_{kt}\e^{\rho t} \right) \1{t<\tau_k},
\end{equation}
where $I_k = (u_k')^{-1}$, and $y_k$ is chosen so that $\xi_{k0}^{-1}\E^k[\int_0^T\xi_{kt}c_{kt}\d t]=w_k$ (this is always possible.) The corresponding wealth process satisfies
\begin{equation} \label{eq:WkPk}
W_{kt} = \frac{1}{\xi_{kt}} \E^k_t \left[ \int_t^T \xi_{ks} c_{ks}\d s\right], \quad 0\le t\le T, \quad \P_k\text{-a.s.}
\end{equation}
The optimizer $c_k$ is unique up to $\P_k\otimes\d t$-a.e.~equivalence.
\end{proposition}

This result has the following consequence, which establishes our previous claim that $\tau_k$ plays the role of the time of bankruptcy of agent~$k$.

\begin{corollary} \label{C:bankrupt}
Let $c_k$ be any optimal consumption plan for the $k$:th agent's optimization problem, and let $W_k$ be the corresponding wealth process. We then have
\[
\lim_{t\uparrow\tau_k} c_{kt} = \lim_{t\uparrow\tau_k} W_{kt} = 0, \quad \P\text{-a.s.} \quad\text{on}\quad\{\tau_k\le T\}.
\]
\end{corollary}

Corollary~\ref{C:bankrupt} shows that the $k$:th agent will become bankrupt as $t$ reaches $\tau_k$, and that his consumption will decrease to zero. This ``should'' imply that the wealth process is absorbed, since no external capital injections take place. This in turn would force the agent to stop consuming and to stop trading. The next result shows that this intuition is correct, if we require that the strategy employed by agent~$k$ be self-financing and admissible with respect to~$\P$, not just~$\P_k$. The result also shows that the optimal consumption plan given by~\eqref{eq:FOCk} indeed has this property.

\begin{proposition} \label{P:nores}
$(i)$ Let $c$ be any $\P_k$-feasible consumption plan that is optimal for agent~$k$. If in addition $c$ is $\P$-feasible, financed by $(\phi,\pi)$, with wealth process $W$, then the following \emph{No Resurrection} property holds $\P\otimes\d t$-a.e.~on $[\tau_k,T]$:
\[
W_t=0, \quad c_t=0, \quad \phi_t=0, \quad \pi_t=0.
\]
$(ii)$ The $\P_k$-feasible consumption plan $c_k$ given in~\eqref{eq:FOCk} is $\P$-feasible.
\end{proposition}

Note that Corollary~\ref{C:bankrupt} shows that there can be no ``positive surprises'' for an agent when his beliefs are proven wrong. That is, if an event that was initially deemed impossible nonetheless occurs, the agent will never see a surprise \emph{increase} in wealth---he will always become insolvent. This fact has nothing to do with equilibrium, but is rather a consequence of the structure of the optimal consumption plan and trading strategy.

\subsection{Equilibrium}

The preceding results show that the $k$:th agent's actions after $\tau_k$ are pinned down in a very natural way: If his consumption is to remain ``globally'' feasible (i.e., feasible under~$\P$), he must cease to consume and to trade. On the other hand, we have also seen that requiring $\P$-feasibility does not affect the optimal behavior before~$\tau_k$. This motivates the following definition of equilibrium.

\begin{definition}[Equilibrium]
An \emph{equilibrium} is a set of price processes $(S_0,S)$ of the form~\eqref{eq:Si} and~\eqref{eq:S0}, together with consumption plans and trading strategies $\{c_k, (\phi_k,\pi_k): k=1,\ldots,K\}$ such that the following conditions hold:
\begin{itemize}
\item[$(i)$] Optimality: For each $k$, $c_k$ is optimal for agent~$k$, $\P$-feasible, and financed by $(\phi_k,\pi_k)$.
\item[$(ii)$] Market clearing: $\phi_1+\cdots+\phi_K=0$, $\pi_1+\cdots+\pi_K=S$, $c_1+\cdots+c_K=D$.
\end{itemize}
\end{definition}

We can now solve for equilibrium using the same well-known procedure as in the classical case with homogeneous beliefs (or heterogeneous but equivalent beliefs). We therefore only give a rough outline. It is always understood that Assumption~\ref{A:max} is satisfied.

Propositions~\ref{P:opt} and~\ref{P:nores} imply that the optimal consumption plans are given by the first order condition~\eqref{eq:FOCk}. Summing over $k$, imposing market clearing, and using the expression~\eqref{eq:xik} for~$\xi_k$, yields
\begin{equation}\label{eq:Dt)}
D_t = \sum_{k\, :\, t<\tau_k} I_k\left( \frac{1}{Z_{kt}} y_k e^{\rho t} \xi_t\right).
\end{equation}
Now define
\[
\Phi(y; \nu_1,\ldots,\nu_K) = \sum_{k\, :\, \nu_k>0} I_k\left(\frac{y}{\nu_k}\right)
\]
and let $x\mapsto \Phi^{-1}(x;\nu_1,\ldots,\nu_K)$ be the inverse of $y\mapsto\Phi(y;\nu_1,\ldots,\nu_K)$. Provided not all~$\nu_k$ are zero, the Inada conditions imply that the inverse exists (see also the discussion in \citep[page~172]{Karatzas/Shreve:1998}.) In view of~\eqref{eq:Dt)} we obtain
\begin{equation} \label{eq:SPD_2}
\xi_t = \Phi^{-1}\left(D_t;  \frac{Z_{kt}}{y_k e^{\rho t}} : k=1,\ldots,K \right).
\end{equation}
As long as the right side is a strictly positive continuous semimartingale with absolutely continuous drift, it defines a valid candidate state price density $\xi_t$. Substituting it back into the agents' budget constraints (which are binding by Proposition~\ref{P:opt}) yields a system of equations for $y_1,\ldots,y_K$, whose solution completely specifies $\xi_t$. This in turn yields the interest rate and market price of risk via~\eqref{eq:SPD}. Existence and uniqueness of $y_1,\ldots,y_K$ can be treated as in the classical case, see for instance~\citep{Karatzas/Shreve:1998} or~\citep{Basak:2000}. In the examples we consider below, the solution is found explicitly. Given the state price density, the equilibrium stock prices can be computed; the result is given in the following proposition.

\begin{proposition} \label{P:prices}
In equilibrium, the stock prices are given by
\begin{equation}\label{eq:S1}
S_{it} = \frac{1}{\xi_t} \E_t\left[ \int_t^T \xi_s D_{is} \d s\right],
\end{equation}
which is automatically of the form~\eqref{eq:Si}. The value of the market portfolio is
\begin{equation}\label{eq:Sbar1}
\overline S_t = \frac{1}{\xi_t} \E_t\left[ \int_t^T \xi_s D_s \d s\right].
\end{equation}
\end{proposition}

At this point we have obtained a candidate state price density~$\xi$ and candidate market prices~$S_i$, as well as optimal wealth processes $W_k$ and consumption plans $c_k$ that satisfy $W_1+\cdots+W_k=\overline S$ and $c_1+\cdots+c_K=D$ by construction. It therefore only remains to verify that the market is complete (i.e., $\sigma_t$ is $\P\otimes\d t$-a.e.~invertible), and that the optimal trading strategies $(\phi_k,\pi_k)$ satisfy market clearing. The question of endogenous completeness is a difficult one, although some results in this direction are available, see~\citep{AndRaym:2008, Hugonnier/Malamud/Trubowitz:2012, Kramkov:2013b}. We do not discuss it further here. As for market clearing of strategies, this is a consequence of market clearing of wealth processes in the presence of market completeness. The well-known proof of this fact is omitted.

Note that we may view the state price density $\xi_t$ as arising from a representative agent with stochastic marginal utility $U'_t(x) = \Phi^{-1}(x; Z_{kt}e^{-\rho t}/y_k: k=1,\ldots,K)$ and beliefs~$\P$. It may or may not be possible to construct new beliefs~$\widetilde \P \sim \P$ and a deterministic representative utility function~$\widetilde U(x)$ that result in the same state price density $\xi_t$. This topic is discussed in~\citep{Jouini/Napp:2007}.

\begin{example}[Logarithmic investors] \label{ex:log}
Suppose $u_k(\cdot)=\log(\cdot)$ for all $k$. In this case one can provide explicit expressions for many equilibrium quantities of interest. As expected, these expressions are identical to the ones obtained in the case of heterogeneous but equivalent beliefs, see~\citep{Basak:2005}. Proofs of the following statements are given in the Appendix.

The ``reference'' state price density is given by
\begin{equation} \label{eq:xi_log}
\xi_t = \frac{1}{D_t\eta(0)} e^{-\rho t} \left( w_1 Z_{1t} + \cdots +  w_K Z_{Kt}\right),
\end{equation}
where we define
\begin{equation} \label{eq:eta}
\eta(t)=\frac{1 - e^{-\rho (T-t)}}{\rho}.
\end{equation}
The equilibrium value of the market portfolio is
\begin{equation} \label{eq:Sbarlog}
\overline S_t = D_t \int_t^T e^{-\rho (s-t)}\d s = D_t \eta(t),
\end{equation}
as is typical in models with logarithmic investors. The equilibrium consumption and wealth processes for the $k$:th agent are given by
\begin{equation} \label{eq:cklog}
c_{kt} = \frac{w_k}{e^{\rho t}\xi_{kt}\eta(0) }, \qquad 0\le t<\tau_k
\end{equation}
and
\begin{equation} \label{eq:Wklog}
W_{kt} = c_{kt}\eta(t), \qquad 0\le t\le T.
\end{equation}
Finally, we consider the interest rate~$r_t$ and market prices of risk~$\theta_t$ and~$\theta_{kt}$. Since $w_1 Z_{1t}+\cdots+w_K Z_{Kt}$ is a positive martingale, there exists an $X$-integrable process $\gamma=(\gamma_t)_{0\le t\le T}$ such that
\begin{equation} \label{eq:gammalog}
w_1 Z_{1t}+\cdots+w_K Z_{Kt}= (w_1+\cdots+w_K) \Ecal\Big( \int \gamma_s^\top \d X_s \Big)_t,
\end{equation}
where $\Ecal(\cdot)_t$ denotes stochastic exponential. In terms of the process~$\gamma$ we have
\begin{equation} \label{eq:rthetalog}
r_t = \rho+a_t-v_t^\top(v_t-\gamma_t)
\qquad\text{and}\qquad
\theta_t = v_t-\gamma_t.
\end{equation}
Moreover, we have
\begin{equation} \label{eq:gammaklog}
Z_{kt}=\Ecal\Big(\int \gamma_{ks}^\top \d X_s\Big)_t
\end{equation}
for some process $\gamma_k=(\gamma_{kt})_{0\le t<\tau_k}$ that satisfies, $\P$-a.s.,
\begin{equation} \label{eq:gammakintlog}
\int_0^t \|\gamma_{ks}\|^2 \d s<\infty \quad \text{for}\quad t<\tau_k, \quad\text{but}\quad \int_0^{\tau_k} \|\gamma_{ks}\|^2 \d s=\infty.
\end{equation}
The individual state price density of the $k$:th agent is then given by
\begin{equation} \label{eq:thetaklog}
\theta_{kt} = \theta_t + \gamma_{kt} = v_t - \gamma_t + \gamma_{kt}, \qquad t<\tau_k.
\end{equation}
\end{example}

\section{Equilibrium bubbles} \label{S:B}

The equilibrium setting with non-equivalent beliefs described above gives rise to bubbles in equilibrium which are \emph{subjective} in the sense that only some agents perceive them to be present. This section discusses this phenomenon. We first introduce the \emph{fundamental value} at time~$t$ of a cash flow $c$: it is the minimal amount needed at~$t$ to replicate that cash flow almost surely over the time interval $[t,T]$ using a self-financing admissible strategy. This definition is standard is a complete market setting, see e.g.~\citep{Loewenstein/Willard:2000, Jarrow/Protter/Shimbo:2006, Heston/Loewenstein/Willard:2007, Hugonnier:2012}. Since different agents may have different nullsets, their notion of fundamental value may differ as well. With this in mind, the fundamental value under $\P_k$ is the minimal time~$t$ wealth for which~$c$ becomes~$\P_k$-feasible over~$[t,T]$. We denote this quantity by $F^k_t(c)$. By Proposition~\ref{P:budget} and Equation~\eqref{eq:locequiv} it is given by
\[
F^k_t(c) = \frac{1}{\xi_{kt}} \E^k_t\left[ \int_t^T \xi_{ks} c_s \d s\right], \quad 0\le t < \tau_k, \quad \P\text{-a.s.}
\]
The fundamental value after $\tau_k$ is left undefined. We denote by $F_t(c)$ the fundamental value computed under the reference measure~$\P$.

The \emph{bubble} associated with a dividend-paying traded asset is defined as the difference between its market price, given endogenously in equilibrium, and its fundamental value. The bubbles on the stocks and the market portfolio, as perceived by the $k$:th agent, are thus given by
\[
B^k_{it} = S_{it} - F^k_t(D_i), \qquad \overline B^k_t = \overline S_t - F^k_t(D).
\]
A similar analysis is valid for the riskless asset. Indeed, it can be viewed as derivative paying $S_{0T}$ at time~$T$ with no intermediate cash flows. By Proposition~\ref{P:budget} and Equation~\eqref{eq:locequiv} its fundamental value and bubble, which we denote by $F^k_{0t}$ and $B^k_{0t}$, respectively, are given by
\[
F^k_{0t} = \frac{1}{\xi_{kt}} \E^k_t\left[ \xi_{kT}S_{0T} \right], \qquad B^k_{0t} = S_{0t} - F^k_{0t}.
\]
A detailed discussion of why the presence of bubbles is in fact consistent with optimal choice and absence of (unlimited) arbitrage is given in \citep[Section~3]{Hugonnier:2012}.

The expressions for the fundamental values and bubbles suggest that different agents could perceive different bubbles. It is in this sense bubbles are \emph{subjective}: they depend on the agent through his beliefs. However, this dependence is confined to the nullsets since, whenever two agents agree about zero probability events, they also agree about the size of the bubbles. This is one consequence of the following theorem, which gives the relation between fundamental values, and hence bubbles, computed under different beliefs.

\begin{theorem} \label{T:BFV}
The fundamental value under $\P$ of a cash flow $c$ can be decomposed as follows, for each $k=1,\ldots K$:
\[
F_t(c) = F^k_t(c) + \frac{1}{\xi_t} E_t\left[ \int_{\tau_k}^T \xi_s c_s \d s \right], \quad 0\le t<\tau_k, \quad \P\text{-a.s.}
\]
\end{theorem}

The interpretation of this decomposition is simple: the net present value calculation conducted by agent~$k$ takes fewer payments into account; indeed, it only includes cash flows prior to $\tau_k$ (which is of course equal to infinity almost surely under~$\P_k$, see~\eqref{eq:Pk0}.) In particular, a traded asset whose dividend stream continues after $\tau_k$ will be perceived as overvalued by agent~$k$. This is (typically) the case both for the risk-free asset and the stocks, implying that their prices will have bubble components. 

The following results are immediate consequences of the theorem. The first one states that each agent will invest in such a way that his portfolio is free from bubbles, from the perspective of his own beliefs. This is as it should be, since the agent would otherwise be throwing away wealth by partially investing in a suicide strategy. The second corollary shows that fewer nullsets means larger perceived bubbles. Finally, the third corollary gives an expression for the equilibrium bubbles on the traded assets.

\begin{corollary}
Each equilibrium wealth process $W_k$ has no bubble component when viewed either under~$\P$ or under~$\P_k$.
\end{corollary}

\begin{corollary}
If $\tau_k\le\tau_\ell$ (i.e.~if $\P_k\ll \P_\ell$), then $B^k_{it}\ge B^\ell_{it}$ for $0\le t<\tau_k$, $\P$-a.s.
\end{corollary}

\begin{corollary} \label{C:bubble}
The bubbles on the individual stocks, the market portfolio, and the riskless asset, as perceived by the $k$:th agent, are given by
\[
B^k_{it} = \frac{1}{\xi_t} \E_t\left[ \int_{\tau_k}^T \xi_s D_{is} \d s \right], \qquad \overline B^k_t = \frac{1}{\xi_t} \E_t\left[ \int_{\tau_k}^T \xi_s D_s \d s \right],
\]
\[
B^k_{0t} = S_{0t}\left(1 -  \E^k_t \left[\frac{\xi_{kT}S_{0T}}{\xi_{kt}S_{0t}} \right] \right).
\]
\end{corollary}

We emphasize that heterogeneous beliefs need not induce bubbles. Rather, bubbles appear when there is disagreement about nullsets. Moreover, while agents can disagree about the size of a bubble, they can still agree that a bubble is present (unless there is only one agent left in the economy): this happens when $\P_k(\tau_k\le T)>0$ for all~$k$, which is of course in no way incompatible with Assumption~\ref{A:max}. Examples of this kind will be discussed below.

\begin{example} \label{ex:bubblelog}
Consider again the setting in Example~\ref{ex:log} with logarithmic investors, and assume that there are only $K=2$ agents in the economy. Corollary~\ref{C:bubble} then shows that the subjective bubbles on the market portfolio are given by
\[
\overline B^k_t = \frac{1}{\xi_t} \E_t\left[ \int_{\tau_k}^T \frac{e^{-\rho s}}{\eta(0)}\left( w_1 Z_{1s} + w_2 Z_{2s}\right) \d s\right] = \frac{w_\ell}{\xi_t\eta(0)} \E_t \left[ \int_{\tau_k}^T e^{-\rho s} Z_{\ell s} \d s\right], \quad t<\tau_k,
\]
where $\ell\in\{1,2\}$, $\ell\ne k$. For the second equality we used that $Z_{ks}=0$ for $s\ge\tau_k$.
\end{example}

\subsection{No Arbitrage conditions and equivalent martingale measures} \label{S:EMM}

We now comment briefly on the question of equivalent martingale measures. In the present setting, an \emph{equivalent local martingale measure (ELMM)} for agent~$k$ is a probability measure $\Q \sim \P_k$ such that the discounted cum-dividend stock price processes,
\[
\frac{S_{it}}{S_{0t}} + \int_0^t \frac{D_{is}}{S_{0s}}\d s, \qquad 0\le t\le T,
\]
are local martingales with respect to~$\Q$. If these processes become true martingales, $\Q$ is called an \emph{equivalent martingale measure (EMM)}. The existence of ELMMs and/or EMMs are closely related to various conditions of no-arbitrage type: existence of an ELMM is equivalent to the condition NFLVR (``No Free Lunch With Vanishing Risk'', see~\citet{Delbaen/Schachermayer:1994,Delbaen/Schachermayer:1998}), while existence of an EMM is equivalent to NFLVR together with Merton's No~Dominance condition (see \citet{Loewenstein/Willard:2000} and \citet{Jarrow:2012fk} for a discussion of this fact and connections to market efficiency.) In the complete market case, it is well understood that the existence of asset pricing bubbles on the stock prices is consistent with NFLVR, but not with No~Dominance.

The situation regarding bubbles on the riskless asset is different. In the complete market setting, the only candidate density process for an ELMM is the local martingale $\xi_{kt}S_{0t}$. Hence, in view of Corollary~\ref{C:bubble}, the presence of a nonzero bubble on the riskless asset precludes $\xi_{kt}S_{0t}$ from being a true martingale, and thus from being the density process of an equivalent probability measure. It follows that no ELMM can exist, and that NFLVR fails. Nonetheless, the market can be in equilibrium. However, since a state price density~$\xi_{kt}$ exists, the weaker condition of ``No Arbitrage of the First Kind'', or ``No Unbounded Profit with Bounded Risk'' are satisfied, see~\citep{Karatzas/Kardaras:2007}. This condition, without which utility maximization cannot be done, is thus also necessary and sometimes sufficient for equilibrium to be possible.

\subsection{Historical probability measures}
From an econometric point of view one would like to single out a particular probability measure $\P^*$ as the \emph{historical} or \emph{objective} probability measure, which at least in principle can be identified via statistical methods using a sufficiently long time series of observations of the underlying economic variables. In our setting the choice of $\P^*$ becomes particularly delicate, as the choice of nullsets is ambiguous. Indeed, the meaning of \eqref{eq:locequiv} is precisely that prior to $\tau_k$, no statistical method can decide whether $\{\tau_k\le T\}$ is a nullset or not.

For this reason, the model can be interpreted in different ways: $\P$ could be taken as the historical measure, but any of the $\P_k$ would also be a valid choice. Depending on which choice is made, different agents will be ``wrong'' in their beliefs. Once $\P^*$ has been chosen, it also becomes possible to classify bubbles as ``illusory'' or ``real''. For example, in a two-agent situation where $\P_1\ll\P_2$, agent~1 perceives a bubble while agent~2 does not. If we take $\P^*=\P_2$, then the bubble perceived by agent~1 is fictitious in the sense that the (limited) arbitrage he perceives is illusory and arises from the failure to account for certain catastrophic scenarios. If, on the other hand, we take $\P^*=\P_1$, then there is indeed limited arbitrage, this time caused by the unnecessarily cautious behavior of agent~2, who hedges against scenarios that will never occur. Illusory arbitrage in the context of performance evaluation is discussed in~\citep{Jarrow/Protter:2013}.

\section{Examples} \label{S:ex}

In this section we give several examples of economies where agents disagree about nullsets, and asset price bubbles are present. We also discuss the nature of the trading strategies that lead to bankruptcy of an agent~$k$ on the event $\{\tau_k\le T\}$. In order to make the examples as explicit as possible, we always consider two agents with logarithmic utilities $u_1(x)=u_2(x)=\log(x)$.

\subsection{One risky asset, the first agent optimistic} \label{S:ex1}
Let us consider an economy with one risky asset, where the first agent is optimistic about the future dividend stream of the asset, while the second agent has ``neutral'' views. The dividend process follows a geometric Brownian motion,
\[
D_t = D_0 \exp\left( vX_t  - \frac{v^2}{2} t \right),
\]
where $X$ is Brownian motion under $\P$, and $D_0>1$, $v>0$. Let $\tau_1$ be the first time $D_t$ hits one,
\begin{equation}\label{eq:ex1tau1}
\tau_1 = \inf\{ t \in [0,T] : D_t = 1\},
\end{equation}
and define
\[
Z_{1t} = \frac{D_{t\wedge\tau_1} - 1}{D_0-1}.
\]
This is a nonnegative martingale starting at $Z_{10}=1$, so we may define the beliefs of agent~1 by $\d\P_1 = Z_{1T}\d\P$. The beliefs of agent~2 are given by $\P_2=\P$. This gives an economy and a beliefs structure that fits into the general framework developed above, where $\P_1\ll \P_2$ holds but equivalence fails. The interpretation of this choice of $\P_1$ is that agent~1 is optimistic in the sense that \emph{he does not believe the dividend process can fall below one}. In equilibrium, therefore, we expect agent~1 to attempt to exploit what he perceives as the unnecessarily cautious behavior of agent~2. In view of Example~\ref{ex:log}, the equilibrium quantities are easily computed:

\begin{proposition} \label{P:ex1}
The equilibrium market prices of risk for the two agents are given by
\[
\theta_{1t} = \theta_t + \frac{vD_t}{D_t-1}, \qquad \theta_{2t} = \theta_t,
\]
where
\begin{equation} \label{eq:MPRex1}
\theta_t = v - \1{t<\tau_1} \frac{vD_t}{D_t - 1 + \frac{w_2}{w_1}(D_0-1)}.
\end{equation}
The equilibrium interest rate is given by
\[
r_t = \rho - v\theta_t,
\]
and the stock price is
\[
S_t = D_t \eta(t),
\]
with $\eta(t)$ as in~\eqref{eq:eta}. In particular, the stock price volatility is $v>0$, so that the market is complete.
\end{proposition}

Note that the ``reference'' market price of risk $\theta_t$ is bounded, which means that the same is true for the second agent's subjective market price of risk $\theta_{2t}$. However, the first (optimistic) agent's market price of risk explodes to $+\infty$ as~$t$ increases to~$\tau_1$. As will be shown below, this has consequences for the investment behavior of agent~1 close to~$\tau_1$.

Next, the sign of $\theta_{2t}$ prior to $\tau_1$ depends on the relative initial wealth of the two agents. Indeed, it follows from~\eqref{eq:MPRex1} that for $t<\tau_1$, $\theta_{2t}>0$ holds if and only if $w_1/w_2 < D_0-1$. Thus, if the optimistic agent dominates the economy from the outset, then the growth potential of the stock price becomes so strong that agent~2 is willing to sustain \emph{negative compensation} (in effect, to pay) for investing in the stock. On the other hand, the interest rate moves in the opposite direction, meaning that investors require high yields in order to put money in the riskless asset. If on the other hand we have $w_1/w_2 > D_0-1$, then $\theta_{2t}<0$.

From the point of view of the optimistic agent~1, things look different. His subjective market price of risk $\theta_{1t}$ is always positive, regardless of the initial wealth distribution. Moreover, agent~1 will view the equilibrium interest rate as too high. To see this, we apply It\^o's formula to derive the dynamics of the state price density of agent~1, given by $\xi_{1t}=\xi_1/Z_{1t}$. The result is
\[
\frac{\d \xi_{1t}}{\xi_{1t}} =  - \int_0^t \left( \rho - v\theta_{1s} \right) \d s + \int_0^t \theta_{1s}^\top \d X_s.
\]
Hence, from the optimistic agent's perspective, the correct interest rate should be $\rho-v\theta_{1t}$. This differs from the interest rate that actually prevails in equilibrium, which satisfies
\[
r_t = \rho - v\theta_t = (\rho - v\theta_{1t})  - \frac{v^2 D_t}{D_t-1}.
\]
The perceived mis-specification of the interest rate explodes as $t$ increases to $\tau_1$, which we interpret as a statement about the relative size of the bubble on the stock and the bubble on the riskless asset, as perceived by agent~1: the bubble on the riskless asset is larger in relative terms. This intuition is supported by the following result, which deals with equilibrium trading strategies.

\begin{proposition} \label{P:ex1_2}
The equilibrium strategy of the $k$:th agent ($k=1,2$) is given by
\[
\left(\!
\begin{array}{c}
\phi_{kt}\\
\pi_{kt}
\end{array}
\!\right)
=
W_{kt} \left(\!
\begin{array}{c}
1-\theta_{kt} / v\\
\theta_{kt} / v
\end{array}
\!\right),
\qquad 0\le t<\tau_k.
\]
Moreover, we have, $\P$-a.s.,
\[
\lim_{t\uparrow\tau_1} \pi_{1t} = - \lim_{t\uparrow\tau_1} \phi_{1t} = \frac{w_1}{w_2}(D_0-1)^{-1}S_{\tau_1} \quad \text{on}\quad \{\tau_1\le T\}.
\]
\end{proposition}

This result clearly shows how agent~1 attempts to exploit the perceived bubbles. Both the stock and the riskless asset carry bubbles, so ideally agent~1 would like to sell both these assets short. However, due to the admissibility requirement of maintaining nonnegative wealth such a strategy is infeasible. Instead, agent~1 short sells one of the assets (in this example, the riskless assets), while maintaining a long position in the other asset (the stock) in order to guarantee admissibility. Such collateralized trades are discussed in detail in \citep[Section~3]{Hugonnier:2012}. The fact that he goes short in the riskless asset and long in the stock confirms the previous intuition that the riskless asset has a larger bubble. It is also consistent with the interpretation of agent~1 as being optimistic about~$D_t$. In fact, Proposition~\ref{P:ex1_2} leads us to refine our interpretation of the beliefs of agent~1: he is optimistic about the performance of the stock dividends, \emph{relative to the performance of the riskless asset}.

Proposition~\ref{P:ex1_2} also shows that while the \emph{proportions} of wealth invested in the two assets explode as~$t$ increases to~$\tau_1$, this behavior is caused by the vanishing denominator. The numerators, i.e.~the \emph{amounts} held in the two assets, are well-behaved in the sense that they converge to something finite. The net value of this limiting portfolio is zero. However, at every instant strictly prior to~$\tau_1$, agent~1 is convinced that his highly levered position will ultimately result in a profit, and for this reason he continues to trade until his wealth reaches zero. The offsetting positions he then holds in the stock (long) and risk-free asset (short) amount to a loan from agent~2 that was used to buy shares of the stock. At $\tau_1$ agent~1 is forced to liquidate this position, effectively handing over his stock holdings to agent~2, thereby closing out the loan.

Finally, let us comment on the question of equivalent martingale measures and absence of arbitrage. Since we clearly have positive bubbles under~$\P_1$, the discussion in Section~\ref{S:EMM} implies that no equivalent local martingale measure can exist relative to~$\P_1$. On the other hand, the market price of risk of agent~2 is bounded, which implies that an equivalent \emph{true} martingale measure exists relative to~$\P_2$.

\subsection{One risky asset, the first agent pessimistic} \label{S:ex2}
The previous example can easily be modified to yield a situation where agent~1 is \emph{pessimistic}: Take $D_0<1$, let $\tau_1$ be defined by~\eqref{eq:ex1tau1}, and set
\[
Z_{1t}= \frac{1-D_{t\wedge\tau_1}}{1-D_0}
\]
as before. Endowing agent~1 with beliefs $\P_1$ given by $\d\P_1=Z_{1T}\d\P$, he will then assign zero probability to the event that the dividend process \emph{rise above one}. Propositions~\ref{P:ex1} and~\ref{P:ex1_2} still hold verbatim; however, $D_0-1$ is now negative rather than positive, as is $D_t-1$ for $t<\tau_1$. Consequently, the market price of risk of the pessimistic agent~1 will now explode to $-\infty$ as $t$ increases to~$\tau_1$. Moreover, writing the expressions for the limiting optimal amounts of agent~1 (see Proposition~\ref{P:ex1_2}) as
\[
\lim_{t\uparrow\tau_1} \pi_{1t} = - \lim_{t\uparrow\tau_1} \phi_{1t} = - \frac{w_1}{w_2}(1-D_0)^{-1}S_{\tau_1} \quad \text{on}\quad \{\tau_1\le T\},
\]
we see that the pessimistic agent will attempt to exploit the (perceived) bubbles by short selling the stock rather than the riskless asset. This is again consistent with our interpretation of the beliefs of the first agent as \emph{pessimism regarding the performance of the stock dividends relative to the riskless asset}.

\subsection{One risky asset, both agents perceive bubbles} \label{S:ex3}
In this example we consider an economy with one risky asset and two agents, where the beliefs structure is such that both agents simultaneously see a bubble. This is achieved by letting agent~1 be optimistic as before, and assuming that agent~2 operates under the belief that large downward movements of the dividend process are impossible. These ``large downward movements'' will be quantified via \emph{relative drawdowns}. An interesting additional feature of the resulting economy is that the subjective bubble perceived initially by agent~1 may burst before agent~2 goes bankrupt.

Now to the specifics. We let $X$ be Brownian motion under~$\P$, $v>0$, $D_0>1$, and define
\[
D_t = D_0 \exp\left( vX_t  - \frac{v^2}{2} t \right),
\]
as well as
\[
\tau_1 = \inf\{ t \in [0,T] : D_t = 1\}, \qquad Z_{1t} = \frac{D_{t\wedge\tau_1} - 1}{D_0-1}, \qquad \d\P_1=Z_{1T}\d\P.
\]
The beliefs of agent~2 are given as follows. The \emph{relative drawdown} of the dividend process is defined by
\[
{\rm rDD}_t= 1 - \frac{D_t}{D_t^*}, \quad \text{where} \quad D^*_t = \max_{0\le s\le t} D_s.
\]
Fix a constant $\kappa\in(0,1)$ and let $\tau_2$ be the first time the relative drawdown reaches~$1-\kappa$. Equivalently, this is the first time the dividend process becomes a fraction $\kappa$ of the level of its running maximum. That is, we have
\begin{equation} \label{eq:ex22_00}
\tau_2 = \inf\{ t\in[0,T]: {\rm rDD}_t = 1-\kappa\} = \inf\{ t\in[0,T]: D_t = \kappa D^*_t \}.
\end{equation}
The beliefs of agent~2 are now given via the corresponding density process.

\begin{lemma} \label{L:ex_22}
Define a process $Z_2$ by
\[
Z_{2t} = \frac{D_{t\wedge\tau_2} - \kappa D^*_{t\wedge\tau_2}}{(1-\kappa)D_0} \left( \frac{D^*_{t\wedge\tau_2}}{D_0}\right)^{\frac{\kappa}{1-\kappa}}.
\]
Then we have
\begin{equation}\label{eq:ex22_22}
\tau_2 = \inf\{t\in[0,T] : Z_{2t}=0\}.
\end{equation}
Moreover, $Z_2$ is a martingale on $[0,T]$ and satisfies
\[
Z_{2t} = 1 + \int_0^t Z_{2s} \1{s<\tau_2} \frac{\d D_s}{D_s - \kappa D^*_s}.
\]
\end{lemma}

We now set $\d\P_2 = Z_{2T}\d\P$. Under these beliefs, agent~2 will view a relative drawdown of $1-\kappa$ or more as impossible.

At this point we should pause and emphasize that $\P(\max(\tau_1,\tau_2)\le T)>0$ holds, in violation of Assumption~\ref{A:max}. As remarked earlier, this is easily remedied by replacing $\P$ by $\widetilde \P=(\P_1+\P_2)/2$, for instance. Moreover, the specific choice of $\widetilde \P$ does not influence the resulting interest rate, agent-specific market prices of risk, or equilibrium stock prices.  The bubble components perceived by each agent are also independent of the choice of $\widetilde \P$. We may therefore carry out all computations as usual using the original measure $\P$, for all times~$t<\max(\tau_1,\tau_2)$. With this in mind, the equilibrium quantities can now be found.

\begin{proposition} \label{P:ex_22}
The equilibrium market prices of risk for the two agents are given by
\[
\theta_{1t} = \theta_t + \frac{vD_t}{D_t-1}, \qquad \theta_{2t} = \theta_t + \frac{v D_t}{D_t - \kappa D^*_t},
\]
where, for $t<\max(\tau_1,\tau_2)$,
\begin{equation} \label{eq:MPRex1}
\theta_t = v - \frac{v D_t}{w_1Z_{1t}+w_2Z_{2t}}\left(\frac{w_1}{D_0-1}\1{t<\tau_1} + \frac{w_2}{(1-\kappa)D_0}\left(\frac{D^*_t}{D_0}\right)^{\frac{\kappa}{1-\kappa}}\1{t<\tau_2} \right)
\end{equation}
The equilibrium interest rate is given by
\[
r_t = \rho - v\theta_t,
\]
and the stock price is
\[
S_t = D_t \eta(t),
\]
with $\eta(t)$ as in~\eqref{eq:eta}. In particular, the stock price volatility is $v>0$, so that the market is complete. The equilibrium trading strategies are given by the expressions in Proposition~\ref{P:ex1_2}, and we have the following limiting holdings in the stock.
\begin{align}
\label{eq:ex_22_pi1}
\lim_{t\uparrow\tau_1} \pi_{1t} &=  \frac{w_1}{w_2}\, \frac{S_{\tau_1}}{Z_{2\tau_1}}\, \frac{1}{D_0-1} &\text{on}\quad \{\tau_1<\tau_2\}, \quad \P\text{-a.s.}\\
\label{eq:ex_22_pi2}
\lim_{t\uparrow\tau_2} \pi_{2t} &=  \frac{w_2}{w_1}\, \frac{S_{\tau_2}}{Z_{1\tau_2}}\, \frac{\kappa}{1-\kappa} \left(\frac{D_0}{D^*_{\tau_2}}\right)^{\frac{\kappa}{1-\kappa}-1}  &\text{on}\quad \{\tau_2<\tau_1\}, \quad \P\text{-a.s.}
\end{align}
\end{proposition}

Qualitatively the behavior is similar to the previous examples: the subjective market prices of risk explode close to the respective bankruptcy time (both to $+\infty$ in this example). The same thing holds for the optimal fractions invested in the stock and in the riskless asset. For each agent~$k$, the optimal fraction held in the riskless asset tends to $-\infty$ close to $\tau_k$, which means that both agents hope to exploit the bubble on the riskless asset. This is consistent with the fact that they both are optimistic about the dividend process: agent~1 thinks it cannot drop below one, while agent~2 believes that a large drop relative to its all-time maximum is impossible.

Let us comment on the subjective stock price bubbles appearing in this example. They are given by
\begin{equation} \label{eq:ex22_44}
B^k_t = \frac{1}{\xi_t}\E_t\left[ \int_{\tau_k}^{T\wedge\max(\tau_1,\tau_2)} \xi_s D_s \d s\right], \quad t<\tau_k, \quad \P_k\text{-a.s.}
\end{equation}
Since we have both $\P(\tau_1<\tau_2)>0$ and $\P(\tau_2<\tau_1)>0$, it follows directly that both agents perceive a nonzero bubble at time~0. Moreover, the bubble perceived by agent~1 will disappear when agent~2 goes bankrupt at time~$\tau_2$. Similarly, the bubble perceived by agent~2 disappears at time~$\tau_1$. However, there is a different, rather undramatic way in which the subjective bubble seen by agent~1 can burst. To see this, consider the stopping time
\[
\sigma = \inf\left\{ t\in[0,T] : D^*_t = \frac{1}{\kappa} \right\}.
\]
If $\sigma$ occurs while both agents are still present in the economy, i.e.~on the event $\{\sigma< \tau_1\wedge\tau_2\}$, we necessarily have $\tau_2\le \tau_1$. Indeed, for $t\ge\sigma$ we have $D_t^*\ge 1/\kappa$, which means that if the dividend process is to fall down to $D_t=1$, it must first pass the level $\kappa D^*_t$, triggering~$\tau_2$. This shows that $\tau_2\le \tau_1$ on $\{\sigma< \tau_1\wedge\tau_2\}$. Consequently, by~\eqref{eq:ex22_44}, we have
\[
B^1_t = 0, \quad t \in [\sigma,T], \quad \P_1\text{-a.s.~on $\{\sigma<\tau_1\wedge\tau_2\}$}.
\]
On the other hand, strictly prior to $\sigma$, both agents see strictly positive bubbles. To summarize, we have shown how a bubble (perceived by one of the agents) may burst at some time strictly prior to all bankruptcy times.

\subsection{Two risky assets, both agents optimistic} \label{S:ex4}
The last example considers an economy with two risky assets, where both agents condition components of both dividend processes to be large. This leads to a situation where they both perceive bubbles on the assets, although they disagree about the size of the bubble at any given moment. Interestingly, their perceptions of the statistical properties of the aggregate bubble may coincide under certain conditions.

We model the aggregate dividend under~$\P$ as a geometric Brownian motion with drift,
\[
D_t = D_0\exp\left\{ v^\top X_t + \Big(a - \frac{1}{2}\|v\|^2 \Big)t \right\},
\]
where $v=(v^1,v^2)\in\R^2$, $a\in\R$, and $X=(X_1,X_2)$ is two-dimensional standard Brownian motion under~$\P$. The individual dividend processes are then modeled as fractions of the aggregate dividend,
\[
D_{it} = \psi_{it}D_t, \qquad i=1,2,
\]
where the fractions $\psi_1$, $\psi_2$ are given by
\[
\psi_{1t} = \psi_{10} + \int_0^t \psi_{1s}(1-\psi_{1s}) v_\psi^\top \d X_s, \qquad \psi_{2t} = 1-\psi_{1t},
\]
for some vector $v_\psi \in \mathbb R^2$ and $\psi_{10}\in (0,1)$. The above stochastic differential equation for $\psi_1$ has a unique strong solution valued in the open unit interval~$(0,1)$, see \citep{Menzly/Santos/Veronesi:2004}. It is therefore guaranteed that the dividend processes are well-defined and strictly positive. In particular, It\^o's formula implies that each $D_i$ satisfies
\begin{equation} \label{eq:ex3_1}
\frac{ \d D_{it}}{D_{it}} = \Big( v + (-1)^{i-1}(1-\psi_{it})v_\psi\Big)^\top \d X_t + \Big( a + (-1)^{i-1}(1-\psi_{it})v^\top v_\psi\Big)\d t.
\end{equation}
To define the two agents' beliefs, we set
\[
\tau_k = \inf\{ t\in [0,T] : X_{kt} = -1\}, \qquad Z_{kt} = 1 + X_{k,\,t\wedge\tau_k}, \qquad k=1,2,
\]
and define $\d \P_1 = Z_{1T}\d \P$ and $\d \P_2 = Z_{2T}\d\P$. An application of Girsanov's theorem gives the following result, which shows how the two agents perceive the dividend processes.

\begin{proposition} \label{P:ex3_1}
The bivariate process
\[
X^1_t =
\left(\!
\begin{array}{c}
X^1_{1t}\\[3mm]
X^1_{2t}
\end{array}
\!\right)
=
\left(\!
\begin{array}{l}
X_{1t} - \int_0^t \frac{1}{1+X_{1s}}\d s\\[3mm]
X_{2t}
\end{array}
\!\right)
\]
is Brownian motion under $\P_1$. Similarly, the bivariate process
\[
X^2_t =
\left(\!
\begin{array}{c}
X^2_{1t}\\[3mm]
X^2_{2t}
\end{array}
\!\right)
=
\left(\!
\begin{array}{l}
X_{1t} \\[3mm]
X_{2t} - \int_0^t \frac{1}{1+X_{2s}}\d s
\end{array}
\!\right)
\]
is Brownian motion under $\P_2$. The dividend processes $D_1$ and $D_2$ satisfy
\begin{align*}
\frac{ \d D_{it}}{D_{it}} &= \Big( v + (-1)^{i-1}(1-\psi_{it})v_\psi\Big)^\top \d X^k_t \\
&\qquad + \Big( a + (-1)^{i-1}(1-\psi_{it})v^\top v_\psi + (v + (-1)^{i-1}(1-\psi_{it})v_\psi)^\top {\bf e}_k \frac{1}{1+X_{kt}} \Big)\d t
\end{align*}
under $\P_k$, where ${\bf e}_k$ is the $k$:th unit vector in $\mathbb R^2$.
\end{proposition}

Comparing the statement of Proposition~\ref{P:ex3_1} with Equation~\eqref{eq:ex3_1}, we see that the drift of $D_1$ and $D_2$ is greater under both $\P_1$ and $\P_2$ than under $\P$, provided the componentwise inequalities
\[
v + v_\psi > 0 \qquad \text{and}\qquad v - v_\psi > 0
\]
are satisfied. In this sense the agents are both optimistic. Note that as in the previous example we have $\P(\max(\tau_1,\tau_2)\le T)>0$ holds, in violation of Assumption~\ref{A:max}. As remarked earlier, we may still work under $\P$, as long as we take care only to consider times~$t<\max(\tau_1,\tau_2)$.

\begin{proposition} \label{P:ex3_2}
The equilibrium market price of risk for the two agents are given by
\[
\theta_{kt} = \theta_t + \frac{1}{1+X_{kt}} {\bf e}_k, \qquad t<\tau_k,
\]
where
\[
\theta_t = v - \frac{1}{w_1( 1 + X_{1t\wedge\tau_1}) + w_2( 1 + X_{2t\wedge\tau_2})}
\left(\!
\begin{array}{c}
\1{t<\tau_1}w_1 \\
\1{t<\tau_2}w_2
\end{array}
\!\right), \qquad t<\max(\tau_1,\tau_2).
\]
The equilibrium interest rate is given by
\[
r_t = \rho + a - v^\top \theta_t.
\]
\end{proposition}

Note that neither the market prices of risk, nor the interest rate are bounded in this example.
%
%

We finish with a result showing that under certain circumstances, both agents may agree about the (unconditional) distribution of the aggregate bubble. In other words, the two agents not only agree about the presence of a bubble on the market portfolio, but also about its statistical properties.

\begin{proposition} \label{P:ex3_law}
Assume $v=(1,1)$ and $w_1=w_2=w$. Then the bubble on the market portfolio perceived by agent~$k$ is given by
\[
\overline B^k_t =  \frac{w}{\xi_t\eta(0)}\int_0^T \P(\tau_k\le s\mid\Fcal_t) e^{-\rho s} \d s, \qquad t<\tau_k,
\]
Moreover, the law of the process $\overline B^1$ under~$\P_1$ coincides with the law of the process~$\overline B^2$ under~$\P_2$.
\end{proposition}

\section{Conclusion} \label{S:concl}

This paper develops a dynamic equilibrium model, whose only departure from the standard paradigm is the beliefs structure, where agents disagree about zero probability events. The first contribution is to show that an equilibrium can exist in such a setting. In particular, we address potential consistency problems related to the fact that an agent should be able to revise his beliefs if an event occurs that was initially thought to be impossible. The resolution originates with the fact that the agent necessarily becomes insolvent at any such time.

The second contribution is to show that asset pricing bubbles arise naturally in this model. The bubbles are subjective in the sense that they are perceived by some, but not necessarily all, agents, and that different agents may attribute different portions of the equilibrium prices to bubbles. All previous models where bubbles occur in equilibrium require portfolio restrictions in addition to a standard solvency constraint. In the present paper no such additional restrictions are imposed; instead the bubbles are caused by the disagreement about nullsets.

Several explicit examples are analyzed in order to illustrate some of the phenomena that can occur. In particular, it is shown how agents attempt to exploit perceived bubbles via collateralized long-short strategies. Moreover, all agents may simultaneously see bubbles, and they can even agree about the unconditional distribution of the bubble on the market portfolio. Bubbles can burst when some agent becomes bankrupt, but also at earlier points in time.

\appendix

\section{Proofs} \label{S:Pf}

The following lemma is a Bayes' rule for non-equivalent probability measures. It is a key mathematical tool used in this paper.

\begin{lemma} \label{L:Bayes}
Consider a probability measure $\widetilde \P\ll \P$, define $Z = \frac{\d \widetilde \P}{\d \P}$, and let $\Gcal\subset\Fcal$ be a sub-$\sigma$-field. Then for any $\widetilde \P$-integrable random variable $Y$ we have
\[
\widetilde \E\left[ Y \mid \Gcal \right] = \frac{1}{\E[Z\mid\Gcal]}\E\left[ Z Y \mid \Gcal \right] \quad\text{on}\quad \{ \E[Z\mid\Gcal]>0\}, \quad \P\text{-a.s.},
\]
where $\widetilde \E[\,\cdot\,]$ denotes expectation under~$\widetilde\P$. In particular, defining $A=\{\E[Z\mid\Gcal]>0\}$, it follows that $\widetilde \E\left[ Y \mid \Gcal \right]\oo_A$ is $\P$-a.s.~uniquely defined.
\end{lemma}

\begin{proof}
The proof is similar as in the case of equivalent measures. Details can be found, for instance, in~\citep[Lemma~12]{Larsson:2013}.
\end{proof}

The proofs of Propositions~\ref{P:budget} and~\ref{P:opt} are standard, and rely on the fact that martingale representation under $\P$ implies martingale representation under $\P_k\ll \P$. An outline of the proofs are given for the sake of completeness.

\begin{proof}[Proof of Proposition~\ref{P:budget}]
The argument proceeds as in the standard case (see \citep[Theorem~3.5]{Karatzas/Shreve:1998}, for instance), once we know that any $\P_k$-martingale can be written as the stochastic integral with respect to $X^k=X-\int Z_s^{-1}\d\langle X,Z_k\rangle_s$, which is $n$-dimensional Brownian motion under~$P_k$. This is true by \citep[Theorem~III.5.24]{Jacod/Shiryaev:2003}.
\end{proof}

\begin{proof}[Proof of Proposition~\ref{P:opt}]
Fix a utility function $u(\cdot)$ and let $\Acal(w)$ be the set of all consumption plans $c\ge 0$ that are $\P$-feasible given initial wealth $w$. Since $\E[\int_0^T \xi_t c_t \d t]\le \xi_0 w$ for any $c\in\Acal(w)$ by Proposition~\ref{P:budget}, we have, for any $y\ge0$,
\begin{align*}
\sup_{c\in\Acal(w)} \E\left[ \int_0^T e^{-\rho t}u_k(c_t)\d t\right]
&\le \sup_{c\in\Acal(w)} \E\left[ \int_0^T e^{-\rho t}u(c_t)\d t\right] - y \left( \E\left[\int_0^T \xi_t c_t \d t\right] - \xi_0 w\right) \\
&= \sup_{c\in\Acal(w)} \E\left[ \int_0^T \left( e^{-\rho t}u(c_t) - y \xi_t c_t\right) \d t\right] + y \xi_0 w \\
&\le \E\left[ \int_0^T \left( e^{-\rho t}u(c^y_t) - y \xi_t c^y_t\right) \d t\right] + y \xi_0 w \\
&= \E\left[ \int_0^T e^{-\rho t}u(c_t^y)\d t\right] - y \left( \E\left[\int_0^T \xi_t c^y_t \d t\right] - \xi_0 w\right),
\end{align*}
where $c^y_t=(u')^{-1}(y\xi_t e^{\rho t})$ is the pointwise maximizer of the function $x\mapsto e^{-\rho t}u(x)-y\xi_t x$. Now choose $y\ge 0$ so that $\E[\int_0^T \xi_t c^y_t \d t]=\xi_0 w$. The Inada conditions imply that this is always possible. Then $c^y\in\Acal(w)$ by Proposition~\ref{P:budget}, and we obtain
\[
\sup_{c\in\Acal(w)} \E\left[ \int_0^T e^{-\rho t}u_k(c_t)\d t\right] \le \E\left[ \int_0^T e^{-\rho t}u(c_t^y)\d t\right].
\]
Hence $c^y$ is optimal. The same argument goes through with~$\P$, $\xi$, and~$u(\cdot)$ replaced by~$\P_k$, $\xi_k$, and~$u_k(\cdot)$, respectively, except that we now define $c^y_t=(u_k')^{-1}(y\xi_{kt}e^{\rho t})\1{t<\tau_k}$. This gives~\eqref{eq:FOCk}. The $\P_k\otimes \d t$-a.e.~uniqueness follows from the strict concavity of the mapping $c\mapsto U_k(c)$. The form~\eqref{eq:WkPk} of the wealth process is a consequence of Proposition~\ref{P:budget} and the optimality of~$c^y$.
\end{proof}

\begin{proof}[Proof of Corollary~\ref{C:bankrupt}]
Since $c_k$ is optimal, the first order condition~\eqref{eq:FOCk} together with the uniqueness assertion in Proposition~\ref{P:opt} yields $c_{kt} = I_k( y_k e^{\rho t}\xi_t / Z_{kt})$ for a.e.~$t<\tau_k$, $\P_k$-a.s. By~\eqref{eq:locequiv} this also holds $\P$-a.s. The statement about $c_k$ now follows because $\min_{0\le t\le T}\xi_t>0$ and $\lim_{t\uparrow\tau_k}Z_{kt}=0$, $\P$-a.s., and $\lim_{y\to\infty}I_k(y)=0$.

Now consider $W_k$. Again using~\eqref{eq:locequiv} we deduce that the equality in~\eqref{eq:WkPk} holds for $0\le t<\tau_k$, $\P$-a.s. Hence for $t<\tau_k$ we have, $\P$-a.s.,
\begin{align}
\nonumber W_{kt} &= \frac{1}{\xi_{kt}} \E^k_t \left[ \int_t^T \xi_{ks}c_{ks}\d s\right] \\
\nonumber &= \frac{1}{\xi_{kt}}  \int_t^T\E^k_t \left[\1{s<\tau_k} \xi_{ks}c_{ks}\right]\d s \\
\nonumber &= \frac{1}{\xi_{kt}}   \int_t^T \frac{1}{Z_{kt}}\E_t \left[Z_{ks}\1{s<\tau_k} \xi_{ks}c_{ks}\right]\d s \\
&= \frac{1}{\xi_t}  \E_t \left[ \int_t^{\tau_k} \xi_sc_{ks}\d s\right].
\label{eq:WkP}
\end{align}
Here the second equality uses $\P_k(s<\tau_k)=1$ and Tonelli's theorem, the third uses Bayes' rule (Lemma~\ref{L:Bayes}), and the last equality uses that $\xi_t = \xi_{kt}Z_{kt}$ for $t<\tau_k$. We deduce
\[
\xi_t W_{kt} =  \E_t \left[ \int_0^{\tau_k} \xi_sc_{ks}\d s \right] - \int_0^t \xi_sc_{ks}\d s, \quad t<\tau_k, \quad \P\text{-a.s.}
\]
Together with the positivity of $\xi_t$ and the fact that all martingales are continuous, this yields the claim about $W_k$.
\end{proof}

\begin{proof}[Proof of Proposition~\ref{P:nores}]
By hypothesis, the self-financing property~\eqref{eq:sfk} holds with $\mu_k$ and $\P_k$ replaced by $\mu$ and $\P$, respectively. It\^o's formula then implies that $\xi_t W_t + \int_0^t \xi_s c_s\d s$ is local martingale, and hence a supermartingale since it is nonnegative. Thus $\xi_t W_t$ is also a supermartingale, and therefore absorbed once it reaches zero (by Corollary~\ref{C:bankrupt} this happens at $\tau_k$.) By positivity of $\xi_t$, the same holds for $W_t$. It then follows that $\int_{\tau_k}^t \xi_s c_s \d s=0$ for all $t\in[\tau_k,T]$, $\P$-a.s., which implies that $c_t$ is zero there, at least up to $\P\otimes\d t$-a.e.~equivalence. Returning to~\eqref{eq:sfk}, we see that $\sigma_t^\top \pi_t=0$, and hence $\pi_t=0$, on $[\tau_k,T]$, $\P\otimes\d t$-a.e. This finally yields $\phi_t = W_t-\pi_t^\top\oo=0$, and part~$(i)$ is proved.

For part~$(ii)$, simply note (using~\eqref{eq:xik} and Proposition~\eqref{P:opt}) that
\[
\E\left[ \int_0^T \xi_t c_{kt} \d t \right] = \int_0^T \E\left[  \xi_t c_{kt}\1{t<\tau_k} \right] \d t = \E^k\left[ \int_0^T \xi_{kt} c_{kt}  \d t \right] = w_k.
\]
Hence $c_k$ is $\P$-feasible (with initial wealth~$w_k$) due to Proposition~\ref{P:budget}.
\end{proof}

\begin{proof}[Proof of Proposition~\ref{P:prices}]
The individual wealth processes are determined by Propositions~\ref{P:opt} and~\ref{P:nores}. Using also the equality~\eqref{eq:WkP} we find that they satisfy
\begin{equation} \label{eq:W1}
W_{kt} = \frac{1}{\xi_t} \E_t\left[ \int_t^T \xi_s c_{ks} \d s\right], \quad 0\le t\le T,\quad \P\text{-a.s.}
\end{equation}
(Recall that $c_{kt}=0$ for $t\ge\tau_k$.) But any equilibrium stock prices must satisfy $W_{1t}+\cdots+W_{Kt}=\overline S_t$, which yields the expression~\eqref{eq:Sbar1} for the market portfolio. Furthermore, $\xi_t\overline S_t + \int_0^t \xi_s D_s \d s$ is a true martingale, and this will imply~\eqref{eq:S1}. Indeed, for each~$i$, $\xi_t S_{it}+\int_0^t \xi_s D_{is}\d s$ is a (nonnegative) local martingale, hence a true martingale since it is dominated by the martingale $\xi_t\overline S_t + \int_0^t \xi_s D_s \d s$. Since also $0\le S_{iT}\le \overline S_T=0$, we deduce~\eqref{eq:S1}. To see that $S_{it}$ is indeed of the form~\eqref{eq:Si}, we write
\[
\xi_t S_{it} + \int_0^t \xi_s D_{is}\d s = \E_t \left[ \int_0^T \xi_s D_{is}\d s\right] = S_0 + \int_0^t \vartheta_s^\top \d X_s
\]
for some $X$-integrable process $\vartheta$ whose existence is guaranteed by the martingale representation theorem. Integrating $1/\xi$ against the left and right sides above and rearranging terms (and using the positivity of $S_i$) leads to an expression of the form~\eqref{eq:Si}.
\end{proof}

\begin{proof}[Proof for Example~\ref{ex:log}]
Since $u_k(\cdot)=\log(\cdot)$ for all $k$ we have $\Phi(\xi; \nu_1,\dots,\nu_K)=\xi^{-1}(\nu_1 + \cdots + \nu_K)$, and hence by~\eqref{eq:SPD_2},
\[
\xi_t = \frac{1}{D_t} e^{-\rho t} \left( \frac{1}{y_1 }Z_{1t} + \cdots + \frac{1}{y_K }Z_{Kt}\right).
\]
This is certainly a nonnegative semimartingale, and it is strictly positive under Assumption~\ref{A:max}. As usual in the logarithmic setting, the constant $y_k$ can be computed explicitly using the requirement that the budget constraint be binding, together with the first order condition. Indeed,
\[
w_k = \E\left[\int_0^T \xi_s c_{ks} \d s\right] = \frac{1}{y_k}\int_0^T e^{-\rho s} \E[Z_{ks}] \d s = \frac{1-e^{-\rho T}}{y_k \rho},
\]
so that $y_k = \eta(0)/w_k$. This gives~\eqref{eq:xi_log}. Furthermore, using the martingale property (under~$\P$) of the $Z_k$, we deduce from~\eqref{eq:Sbar1} and~\eqref{eq:xi_log} that
\[
\overline S_t = D_t \int_t^T e^{-\rho (s-t)}\d s = D_t \eta(t),
\]
which is~\eqref{eq:Sbarlog}. Equation~\eqref{eq:cklog} follows directly from Proposition~\ref{P:opt}. Together with \eqref{eq:WkPk} this yields, for $t<\tau_k$,
\[
W_{kt} = \frac{1}{\xi_{kt}} \E^k_t\left[ \int_t^T \frac{1}{y_k e^{\rho s}} \d s \right] = c_{kt}\eta(t).
\]
Note that this expression is valid also for $t\ge\tau_k$ due to the No~Resurrection property, which gives~\eqref{eq:Wklog}. To prove~\eqref{eq:rthetalog}, first note that the aggregate dividend is given by
\[
D_t = D_0 \Ecal\Big( \int a_s \d s + \int v_s^\top \d X_s \Big)_t,
\]
which implies
\[
\frac{1}{D_t} = \frac{1}{D_0} \Ecal\Big( - \int ( a_s - \|v_s\|^2) \d s - \int v_s^\top \d X_s \Big)_t.
\]
Therefore, by Yor's formula,
\begin{align*}
\xi_t &= \frac{w_1+\cdots+w_K}{D_0 \eta(0)}e^{-\rho t}\Ecal\Big( - \int ( a_s - \|v_s\|^2) \d s - \int v_s^\top \d X_s \Big)_t\Ecal\Big( \int \gamma_s^\top \d X_s \Big)_t \\
&= \frac{w_1+\cdots+w_K}{D_0 \eta(0)}\Ecal\Big( - \int ( \rho + a_s - \|v_s\|^2+v_s^\top\gamma_s) \d s - \int (v_s-\gamma_s)^\top \d X_s \Big)_t.
\end{align*}
From this expression we simply read off~$r_t$ and~$\theta_t$. Finally, \eqref{eq:gammaklog}--\eqref{eq:gammakintlog} follow from the fact that $Z_k$ is a nonnegative martingale that reaches zero at~$\tau_k$. Equation~\eqref{eq:thetaklog} then follows from It\^o's formula together with the equality $\xi_{kt}=\xi_t/Z_{kt}$, $t<\tau_k$.
\end{proof}

\begin{proof}[Proof of Theorem~\ref{T:BFV}]
For $t<\tau_k$, write
\[
F_t(c) = \frac{1}{\xi_t} \E_t\left[ \int_t^T \xi_s c_s \d s\right] = \frac{1}{\xi_t} \E_t\left[ \int_t^{\tau_k} \xi_s c_s \d s\right] + \frac{1}{\xi_t} \E_t\left[ \int_{\tau_k}^T \xi_s c_s \d s\right].
\]
The same calculation as in~\eqref{eq:WkP} shows that the first term on the right side equals~$F^k_t(c)$. This proves the theorem.
\end{proof}

\begin{proof}[Proof of Proposition~\ref{P:ex1}]
The results follow from Example~\ref{ex:log} after finding the processes $\gamma$ and $\gamma_k$ appearing in  \eqref{eq:gammalog} and \eqref{eq:gammaklog}. This is easily achieved via It\^o's formula applied to the ``reference'' state price density
\[
\xi_t = \frac{1}{D_t\eta(0)}e^{-\rho t} \left( w_1 \frac{D_{t\wedge\tau_1}-1}{D_0-1} + w_2 \right),
\]
obtained via~\eqref{eq:xi_log}.
\end{proof}

\begin{proof}[Proof of Proposition~\ref{P:ex1_2}]
The optimal wealth process for agent~$k$ is given by~\eqref{eq:cklog} and~\eqref{eq:Wklog}. It\^o's formula and the fact that $\sigma_t=v$ imply that $\pi_{kt}/W_{kt}=\theta_{kt}/v$, and hence $\phi_{kt}=1-\theta_{kt}/v$, for $t<\tau_k$. To prove the statement about the limit of $\pi_{1t}$ and $\phi_{1t}$, first note that
\[
\lim_{t\uparrow\tau_1} \pi_{1t} = \lim_{t\uparrow\tau_1} \left( W_{1t}\frac{ \theta_t}{v} + W_{1t}\frac{D_t}{D_t-1}\right) = \lim_{t\uparrow\tau_1} \frac{W_{1t}}{D_t-1}D_t,
\]
since $\theta_t$ is bounded and $\lim_{t\uparrow \tau_1}W_{1t}=0$. We now compute the limit on the right side. To this end, note that~\eqref{eq:Sbarlog}, \eqref{eq:Wklog}, \eqref{eq:cklog}, and~\eqref{eq:xi_log} imply that for $t<\tau_1$, we have
\begin{equation} \label{eq:ex__WZ}
\frac{W_{1t}}{S_t} = \frac{c_{kt}}{D_t} = \frac{w_1 Z_{1t}}{\eta(0) e^{\rho t}\xi_t D_t} = \frac{w_1 Z_{1t}}{w_1Z_{1t}+w_2Z_{2t}}.
\end{equation}
Rearranging this equation and using the definition of $Z_1$ and $Z_2$ yields
\[
S_t - W_{1t} = \frac{W_{1t}}{Z_{1t}}\, \frac{w_2}{w_1} Z_{2t} = \frac{W_{1t}}{D_t-1}\frac{w_2}{w_1}(D_0-1).
\]
Taking the limit as $t\uparrow\tau_1$ we deduce
\[
\lim_{t\uparrow\tau_1} \frac{W_{1t}}{D_t-1} = \frac{w_1}{w_2} (D_0-1)^{-1}S_{\tau_1},
\]
which then yields the result since $D_{\tau_1}=1$.
\end{proof}

\begin{proof}[Proof of Lemma~\ref{L:ex_22}]
It is clear from~\eqref{eq:ex22_00} that~\eqref{eq:ex22_22} holds. Furthermore, we have
\[
Z_{2t} \le \frac{D^*_{t\wedge\tau_2} - \kappa D^*_{t\wedge\tau_2}}{(1-\kappa)D_0} \left( \frac{D^*_{t\wedge\tau_2}}{D_0}\right)^{\frac{\kappa}{1-\kappa}} = \left( \frac{D^*_{t\wedge\tau_2}}{D_0}\right)^{\frac{1}{1-\kappa}} \le \left( \frac{D^*_T}{D_0}\right)^{\frac{1}{1-\kappa}}.
\]
Since $D$ is a positive martingale, the above bound together with Doob's $L^p$-inequality for $p=1/(1-\kappa)$ yields
\begin{equation} \label{eq:ex22_33}
\E\Big[ \sup_{0\le t\le T} Z_{2t} \Big] \le D_0^{-p} \E\Big[ (D^*_T)^p\Big]
\le D_0^{-p} \left(\frac{p}{p-1}\right)^p \E\left[ D_T^p \right] <\infty.
\end{equation}
The right side is finite since $D_T$ is log-normally distributed and hence have finite moments. We can now prove that $Z_2$ is a martingale, which will complete the proof of the lemma. To this end, we will use \citep[Lemma~2.4]{CherNikeghPlaten:2012} (see also Equation~(4.1) in the same reference), which states that
\begin{equation}\label{eq:ex22_cher}
{\rm rDD}_t = - \int_0^t \frac{\d D_s}{D^*_s} + \log D^*_t - \log D_0.
\end{equation}
Now, define two processes
\[
Y_t = \frac{1}{1-\kappa}\left( 1-\kappa - {\rm rDD}_{t\wedge\tau_2} \right), \qquad \Lambda_t = \frac{1}{1-\kappa}\log\left(\frac{D^*_{t\wedge\tau_2}}{D_0}\right).
\]
Using the definition of ${\rm rDD}_t$ one finds that
\[
Z_{2t} = e^{\Lambda_t} Y_t.
\]
Moreover, \eqref{eq:ex22_cher} yields
\[
Y_t = 1 + \frac{1}{1-\kappa} \int_0^t\1{s<\tau_2} \left( \frac{\d D_s}{D^*_s} - \d \log D^*_s\right),
\]
so by the product rule we get
\begin{align*}
Z_{2t} &= 1 + \int_0^t Z_{2s} \1{s<\tau_2}\left(  \d \Lambda_s + \frac{\d Y_s}{Y_s} \right)\\
&= 1 + \frac{1}{1-\kappa} \int_0^t  Z_{2s} \1{t<\tau_2}\left(  \d \log D^*_s - \frac{\d \log D^*_s}{Y_s}  + \frac{\d D_s}{Y_s D^*_s} \right).
\end{align*}
But $\d D^*_t$, and hence $\d \log D^*_t$, only charges the set $\{t:D_t=D^*_t\}$, and on this set we have $Y_t=1$. Hence the first two terms in the parentheses above cancel, and we arrive at
\[
Z_{2t} = 1 + \int_0^t Z_{2s} \1{s<\tau_2} \frac{\d D_s}{(1-\kappa)Y_s D^*_s}
= 1 + \int_0^t Z_{2s} \1{s<\tau_2} \frac{\d D_s}{D_s - \kappa D^*_s}.
\]
It follows that $Z_2$ is a local martingale, and hence a true martingale in view of~\eqref{eq:ex22_33}. The lemma is proved.
\end{proof}

\begin{proof}[Proof of Proposition~\ref{P:ex_22}]
The proof is similar to the proof of Proposition~\ref{P:ex1_2}, so we omit most details. Let us only indicate how to obtain the limits~\eqref{eq:ex_22_pi1} and~\eqref{eq:ex_22_pi2}. Starting from~\eqref{eq:ex__WZ} (see the proof of Proposition~\ref{P:ex1_2}) we obtain
\[
\lim_{t\uparrow\tau_1} \frac{W_{1t}}{Z_{1t}} = \lim_{t\uparrow\tau_1} \frac{w_1}{w_2}\, \frac{S_t-W_{1t}}{Z_{2t}} = \frac{w_1}{w_2}\, \frac{S_{\tau_1}}{Z_{2\tau_1}} \quad \text{on}\quad \{\tau_1<\tau_2\},
\]
and similarly
\[
\lim_{t\uparrow\tau_2} \frac{W_{2t}}{Z_{2t}} = \frac{w_2}{w_1}\, \frac{S_{\tau_2}}{Z_{1\tau_2}}  \quad \text{on}\quad \{\tau_2<\tau_1\}.
\]
Furthermore, we have the equalities
\[
\frac{W_{1t}D_t}{D_t-1} = \frac{W_{1t}}{Z_{1t}}\, \frac{D_t}{D_0-1}, \qquad t<\tau_1,
\]
and
\[
\frac{W_{2t}D_t}{D_t-\kappa D^*_t} = \frac{W_{2t}}{Z_{2t}}\, \frac{D_t}{(1-\kappa)D_0} \left(\frac{D_0}{D^*_t}\right)^{\frac{\kappa}{1-\kappa}}, \qquad t<\tau_2.
\]
Combining these expressions, and using that $\lim_{t\uparrow\tau_k}W_{kt}\theta_t=0$ for $k=1,2$, we deduce
\[
\lim_{t\uparrow\tau_1} \pi_{1t} = \lim_{t\uparrow\tau_1} W_{1t}\frac{\theta_{1t}}{v} = \lim_{t\uparrow\tau_1} \frac{W_{1t}D_t}{D_t-1} =  \frac{w_1}{w_2}\, \frac{S_{\tau_1}}{Z_{2\tau_1}}\, \frac{D_{\tau_1}}{D_0-1}
\]
on $\{\tau_1<\tau_2\}$, as well as
\[
\lim_{t\uparrow\tau_2} \pi_{2t} = \lim_{t\uparrow\tau_2} W_{2t}\frac{\theta_{2t}}{v} = \lim_{t\uparrow\tau_2} \frac{W_{2t}D_t}{D_t-\kappa D^*_t} =  \frac{w_2}{w_1}\, \frac{S_{\tau_2}}{Z_{1\tau_2}}\, \frac{D_{\tau_2}}{(1-\kappa)D_0} \left(\frac{D_0}{D^*_{\tau_2}}\right)^{\frac{\kappa}{1-\kappa}}
\]
on $\{\tau_2<\tau_1\}$. Using also that $D_{\tau_1}=1$ and $D_{\tau_2}=\kappa D^*_{\tau_2}$ gives~\eqref{eq:ex_22_pi1} and~\eqref{eq:ex_22_pi2}.
\end{proof}

\begin{proof}[Proof of Proposition~\ref{P:ex3_1}]
This is a straightforward application of Girsanov's theorem, see \citep[Theorem~III.41]{Protter:2005}.
\end{proof}

\begin{proof}[Proof of Proposition~\ref{P:ex3_2}]
Just as the proofs of Propositions~\ref{P:ex1} and~\ref{P:ex_22}, this is an application of the expressions given in Example~\ref{ex:log}. We omit the details.
\end{proof}

\begin{proof}[Proof of Proposition~\ref{P:ex3_law}]
The expression in Example~\ref{ex:bubblelog} and the fact that $Z_{kt}=1+X_{k,t\wedge\tau_k}$ yield
\[
\overline B^k_t = \frac{w_\ell}{\xi_t\eta(0)} \E_t \left[ \int_{\tau_k}^T e^{-\rho s} (1+X_{\ell,s\wedge\tau_\ell}) \d s\right], \quad t<\tau_k,
\]
where $\ell\ne k$. Since $X_1$ and $X_2$ are independent, we have
\begin{align*}
\E_t \left[ \int_{\tau_k}^T e^{-\rho s} (1+X_{\ell,s\wedge\tau_\ell}) \d s\right]
&=  \int_0^T \E_t \left[ \1{\tau_k\le s} e^{-\rho s} (1+X_{\ell,s\wedge\tau_\ell})\right] \d s \\
&=  \int_0^T \P(\tau_k\le s\mid\Fcal_t) e^{-\rho s} \d s,
\end{align*}
which gives the claimed expression for~$\overline B^k_t$. Next, again by the independence of $X_1$ and $X_2$, we have that $\P(\tau_k\le s\mid\Fcal_t)=F(s;t,X_k)$ for some functional $F$ that depends on the path of its last argument up to time~$t$. Moreover, $\xi_t$ depends on $X_1$ and $X_2$ only through $X_{1,t\wedge\tau_1}+X_{2,t\wedge\tau_2}$. Since this expression is unaffected by a permutation of the indices~$1$ and~$2$, it follows from Proposition~\ref{P:ex3_1} that the joint law of $\{X_{1,t\wedge\tau_1}+X_{2,t\wedge\tau_2}, X_{1t} : t\in[0,T]\}$ under $\P_1$ coincides with the joint law of $\{X_{1,t\wedge\tau_1}+X_{2,t\wedge\tau_2}, X_{2t} : t\in[0,T]\}$ under $\P_2$. This finishes the proof of the proposition.
\end{proof}

\bibliographystyle{plainnat}

\bibliography{bibl}


\end{document}